\documentclass[%
 aip,
 amsmath,amssymb,
 reprint,%
]{revtex4-1}

\usepackage{amsthm,amsmath,stmaryrd,bbm,hyperref,geometry,color}
\usepackage[utf8]{inputenc}
\usepackage{amssymb}
\usepackage[english]{babel}
\usepackage{graphicx}
\usepackage{amsfonts,amssymb}
\usepackage{verbatim}
\usepackage{enumitem}

\newcommand{\po}{\left(}
\newcommand{\pf}{\right)}
\newcommand{\co}{\left[}
\newcommand{\cf}{\right]}
\newcommand{\cco}{\llbracket}
\newcommand{\ccf}{\rrbracket}
\newcommand{\R}{\mathbb R} 
 
\newcommand{\Z}{\mathbb Z} 
 
\newcommand{\N}{\mathbb N} 
\newcommand{\dd}{\text{d}}
\newcommand{\na}{\nabla}

\newtheorem{thm}{Theorem}

\begin{document}
\preprint{AIP/123-QED}
\title[Velocity jump processes]{Velocity jump processes : an alternative to multi-timestep methods for faster and accurate molecular dynamics simulations}
\author{Pierre Monmarché} 
\email{pierre.monmarche@sorbonne-universite.fr}
\affiliation{Sorbonne Université, Laboratoire Jacques-Louis Lions, UMR 7589 CNRS, and  Laboratoire de Chimie Théorique, UMR 7616 CNRS, F-75005, Paris, France}
\author{Jérémy Weisman}\affiliation{Sorbonne Université, Laboratoire de Chimie Théorique, UMR 7616 CNRS, F-75005, Paris, France }
\author{Louis Lagardère}
\email{louis.lagardere@sorbonne-universite.fr}
\affiliation{Sorbonne Université, Laboratoire de Chimie Théorique, UMR 7616 CNRS, and Institut Parisien de Chimie Physique et Théorique, FR2622 CNRS, F-75005, Paris, France}
\author{Jean-Philip Piquemal}
\email{jean-philip.piquemal@sorbonne-universite.fr}
\affiliation{Sorbonne Université, Laboratoire de Chimie Théorique, UMR 7616 CNRS, and Institut Universitaire de France, F-75005, Paris, France.}

\begin{abstract}
We propose a new route to accelerate molecular dynamics through the use of velocity jump processes allowing for an adaptive time-step specific to each atom-atom pair (2-body) interactions. We start by introducing the formalism of the new velocity jump molecular dynamics, ergodic with respect to the canonical measure. We then introduce the new BOUNCE integrator that allows for long-range forces to be evaluated at random and optimal time-steps, leading to strong savings in direct space. The accuracy and computational performances of a first BOUNCE implementation dedicated to classical (non-polarizable) force fields is tested in the cases of pure direct-space droplet-like simulations and of periodic boundary conditions (PBC) simulations using Smooth Particule Mesh Ewald. An analysis of the capability of BOUNCE to  reproduce several condensed phase properties is provided. Since electrostatics and van der Waals 2-body contributions are evaluated much less often than with standard integrators using a 1fs timestep, up to a 400 \% direct-space acceleration is observed.  Applying the reversible reference system propagator algorithms (RESPA(1)) to reciprocal space (many-body) interactions allows BOUNCE-RESPA(1) to maintain large speedups in PBC while maintaining precision. Overall, we show that replacing the BAOAB integrator by the BOUNCE adaptive framework preserves a similar accuracy and leads to significant computational savings. 

\end{abstract}
\maketitle

\section{Introduction}
Molecular dynamics (MD) is a popular tool that allows the  simulation of complex molecular systems, ranging from materials to biomolecules, mainly used to compute properties by sampling a defined ensemble. In practice, this means that one would like to perform very long trajectories to maximize this sampling. This can be achieved by brute force high performance computing using advanced massively parallel simulation softwares\cite{tinkerhp,AMBERsoftware,GROMACS,NAMD,PLIMPTON19951,DOMDEC,GENESIS,ShawmilliANTON}. However, at the age of exascale computing, any algorithmic enhancements on the statistical physics side would save millions of hours on supercomputers and therefore energy resources. In this connection the development of more efficient MD integrators is an intense field of research. In recent years, lots of mathematical work has been performed, especially in the framework of the Langevin dynamics, and new techniques  emerged such as Leimkuhler's BAOAB \cite{BAOAB,BAOAB2} offering stable, accurate and well understood integration scheme. Of course, to speed up simulations, one would like primarily to be able to use  time-steps as large as possible. In practice, multi-timestep approaches \cite{tuckerman1992reversible}, that are now standard in MD, do offer a substantial acceleration through the use of frequency-driven splittings. However, it comes at a price as resonance effects limit the maximum usable time-step at a given accuracy.\cite{Skeelresonance1,Skeelresonance2}. 
Various alternative strategies have been proposed to enable the use of very large time-steps such as in the Generalized Langevin Equation (GLE)\cite{GLE} or the stochastic isokinetic extended phase-space algorithm \cite{isokin,isokinpol,isokin3}. They are promising but rely either on some empirical fitting\cite{GLE} or have an important impact on the dynamical properties such as the diffusion coefficient\cite{isokin3}, that is an indication of a limitation of the sampling rate. In this last case, the computational gain is reduced by the fact that a longer trajectory is necessary to keep the same quality of sampling. To speed up molecular dynamics without resorting to some fitting and while not affecting too much the dynamics\cite{BerendsenHMR,GLE}, one possibility is to combine well-chosen integrators within a multi-split approach like for BAOAB-RESPA1 to push forward the stability limit \cite{pushing}.\\
This paper proposes to look outside the box and gives an alternative to multi-time-step approaches. To do so, we will explore the possibility to take into account the slowly varying, bounded part of the atomistic potential at play by a velocity jump mechanism, combining classical MD with probabilistic thinning methods and introducing a new integrator:  BOUNCE. The algorithm can somehow be thought of as a multi-time-step integrator where the long-range forces are evaluated at random time-steps. We will first introduce the mathematical framework of the velocity jumps processes, then present the new BOUNCE integrator. Finally, we will evaluate its computational gain and assess its accuracy using classical, non-polarizable force fields, through a set of numerical experiments performed in the framework of a first implementation of BOUNCE within the Tinker--HP software\cite{tinkerhp}, where various properties will be computed and compared to the state-of-the-art integrators.

\section{Method}
\subsection{Building a new molecular dynamics integrators: requirements}
The main goal of MD is to compute expectations with respect to $\mu$ the Boltzmann-Gibbs measure, namely the probability distribution on $\R^{6N}$ with density proportional to $\exp(-\beta H(q,p))$ where $\beta=1/(k_BT)$ is the inverse temperature, $q$ and $p$ are respectively the positions and momenta of $N$ atoms and the Hamiltonian is $H(q,p)=E_{pot}(q) + E_{kin}(p)$. Denoting $M$ the mass matrix of the system, the kinetic energy is $E_{kin}(p) = \frac12 p' M^{-1} p$, where $x'$ denotes the transpose of $x$. The potential energy $E_{pot}$ will be discussed below.

We need to design a (theoretical, at first) trajectory $t\mapsto (q(t),p(t))$ that is ergodic with respect to $\mu$, in the sense that
\begin{equation}\label{EqErgodic}
\frac1t\int_0^t \varphi\po q(s),p(s)\pf \dd s \ \underset{t\rightarrow\infty}\longrightarrow \int \varphi(q,p)\mu(\dd q,\dd p)    
\end{equation}
for all observables $\varphi$ (in some class of functions on $\R^{6N}$), and then to approximate the continuous-time dynamics by a numerical scheme $(q_k,p_k)_{k\in\N}$, so that
\[\frac1K \sum_{k=1}^K \varphi\po q_k,p_k\pf   \ \underset{k\rightarrow\infty}\longrightarrow \int \varphi(q,p)\mu_{\delta}(\dd q,\dd p)\]
where $\mu_\delta$ is an approximation of $\mu$, depending on the time-step $\delta$ of the numerical scheme.

A classical process that is ergodic with respect to $\mu$ is the Langevin diffusion, solution of the stochastic differential equation (SDE)
\[\left\{\begin{array}{rcl}
\dot  q(t) & = & M^{-1}  p(t)\\
\dot p(t) & = &  -\nabla E_{pot}\po q(t)\pf   - \gamma M^{-1} \tilde p(t)  - \sqrt{2\beta^{-1}\gamma }  \dot w(t)\,,
\end{array} \right.\]
where $(w(t))_{t\geqslant 0}$ is a Brownian motion on $\R^{6N}$, so that $\dot w(t)$ is a white noise force, and $\gamma>0$. Rather than entering into technical
details about diffusion processes, SDEs and stochastic calculus in general, let us simply say that this process is the limit as $\delta$ vanishes of the following Euler scheme with timestep $\delta$:
\[\left\{\begin{array}{rcl}
q_{k+1} - q_k & = & \delta M^{-1}  p_k\\
p_{k+1} - p_k & = &  -\delta \nabla E_{pot}\po q_k\pf   - \delta \gamma M^{-1} p_k  + \sqrt{2\beta^{-1}\gamma \delta }  W_k\,,
\end{array} \right.\]
where $(W_k)_{k\in\N}$ is a sequence of independent random variables of dimension $3N$ distributed according to the standard (i.e. mean zero, variance the identity matrix) Gaussian distribution. In fact, rather than this naive Euler scheme, higher order methods are used in practice, as we will see below. Similarly to the Euler scheme, these schemes require at each step the computation for a configuration $q$ of the forces $-\nabla E_{pot}(q)$. The rest of this work is based on a decomposition of the forces of the form
\[\nabla_q E_{pot}(q) \ = \ \sum_{i=0}^K F_i(q)\]
for $K\geqslant 1$ and some vector-fields $F_i$. Suppose that $F_0$ gathers short-range forces, typically cheap to compute (since each particle only interacts with a small number of neighbours through these forces) but fast-varying and high (in molecular dynamics, short-range forces are repulsive and singular at 0, for instance Lennard-Jones forces scale as $1/r^{13}$ where $r$ is the distance between the centres of two atoms), while the $F_i$'s for $i\geqslant 1$ are larger-range forces, numerically more intensive than $F_0$ (since an atom basically interacts with all the others through these forces), but also smaller and more regular. A natural idea is thus not to compute the long-range forces at each step or, in other words, to use different time-steps for the different forces. Indeed, by classical considerations on the convergence of Euler schemes, the time-step should be related to the norm of the forces. This idea leads to the so-called multi-time-step methods, see \cite{TuckermanRossiBerne,GibsonCarter} and references within.
\subsection{Introduction to jump processes}
The method developed in this work is different. We will design a continuous-time process based on the Langevin diffusion
\begin{equation}\label{Eq-SDE-Langevin-F0}
\left\{\begin{array}{rcl}
\dot  q(t) & = & M^{-1}  p(t)\\
\dot p(t) & = &  -F_0\po q(t)\pf   - \gamma M^{-1} \tilde p(t)  - \sqrt{2\beta^{-1}\gamma }  \dot w(t)\,,
\end{array} \right.
\end{equation}
to which random jumps (or collisions, or bounces) will be added that will take into account the forces $F_i$'s, $i\geqslant 1$, in such that a way that ergodicity with respect to $\mu$ is enforced. Then a unique time-step will be used for the discretization. The resulting discrete-time chain will be such that the forces $F_i$ are only computed at random steps.

Before giving a rigorous definition of the process, let us give a brief and informal definition of the jump mechanisms. Denote $v=M^{-1}p$ the velocities of the particles. The trajectory basically follows \eqref{Eq-SDE-Langevin-F0} except that, between times $t$ and $t+\delta$ for $\delta$ small, the process has a probability $\lambda_i(q(t),v(t)) \delta + o(\delta)$ to undergo a collision due to the $i^{th}$ force for all $i\in\cco 1,K\ccf$, where $\lambda_i$ is called the $i^{th}$ jump rate and, here, is given by
\begin{equation}\label{Eq-Lambda_i}
\lambda_i(q,v) \ = \ \beta \max\po 0,   v' F_i(q) \pf\,.
\end{equation}
 At such a collision, the velocity jumps from $v$ to
\begin{equation}\label{Eq-R_i}
R_i(q,v) \ = \ v - 2 \frac{v'F_i(q)}{F_i(q)' M^{-1} F_i(q)}M^{-1} F_i(q)\,.
\end{equation}
Note that this is not defined if $F_i(q)=0$, but in that case the probability to see a collision is zero. When $M$ is a scalar matrix, this means $v$ is reflected orthogonally  to $F_i(q)$. More generally, $v_* = R_i(q,v)$ is the only vector that conserves the kinetic Energy (i.e. such that $v_*'Mv_*=v'Mv$)  such that $v_*'F_i(q) = - v'F_i(q)$. In other words, $M^{1/2}v_*$ is the orthogonal reflection of $M^{1/2}v$ with respect to $M^{-1/2}F_i(q)$.

This kind of velocity jump mechanism have received interest over the past years in various communities concerned with sampling problems\cite{PetersdeWith,MonmarcheRTP,DoucetPDMCMC}. The novelty in this work is that it is integrated within a classical Langevin diffusion through a force splitting, and implemented in a practical MD code.

\section{Theoretical settings}

\subsection{The theoretical continuous-time process}\label{Sec-Deftheorique}

 Let $T_0=0$ and suppose by induction that $(q_t,p_t)$ has been defined for $t\in[0,T_n]$ for some $n\in\N$. We call $T_n$ the $n^{th}$ jump (or collision) time. Let $(\tilde p,\tilde q)_{t\geqslant T_n}$ be the solution of the SDE \eqref{Eq-SDE-Langevin-F0} with initial condition $\po\tilde p(T_n),\tilde q(T_n)\pf =  \po  p(T_n),  q(T_n)\pf$. Let $(\mathcal E_i)_{i\in\cco 1,K\ccf}$ be independent random variables with standard exponential distribution, i.e. such that $\mathbb P(\mathcal E_i>s) = \exp(-s)$ for all $s>0$, independent from $(w(t))_{t\geqslant 0}$. For $i\in\cco 1,K\ccf$, let
 \begin{equation}\label{Eq-Si}
S_i\ = \ \inf \left\{t>T_n\,,\,\mathcal  E_i < \int_{T_n}^t \lambda_i(\tilde q_s,M^{-1} \tilde p_s)\dd s\right\}\,,
 \end{equation}
where $\lambda_i$ is given by \eqref{Eq-Lambda_i}. The next jump time is then defined as $T_{n+1} = \min\{S_i,\ i\in\cco 1,K\ccf\}$. Let $I_{n+1}$ be the index such that $T_{n+1} = S_{I_{n+1}}$ (since almost surely $T_i\neq T_j$ for $i\neq j$, $I_{n+1}$ is uniquely defined). For $t\in[T_n,T_{n+1})$, set $(p(t),q(t))=(\tilde q(t),\tilde p(t))$. At time $T_{n+1}$, the position is continuous, i.e. $q(T_{n+1})=\tilde q(T_{n+1})$, while the velocity $v = M^{-1} p$ undergoes a collision due to the $I_{n+1}^{th}$ jump mechanism, and jumps to the value
\[v\po T_{n+1}\pf \ = \  R_{I_{n+1}}\po \tilde q\po T_{n+1}\pf,  \tilde v\po T_{n+1}\pf\pf   \,,\]
with $R_i$ given by \eqref{Eq-R_i}. The process is thus defined up to $T_{n+1}$ and, by induction, up to $T_k$ for all $k\in\N$. From now on, we assume that  $\|F_i\|_\infty := \sup_{q\in\R^{3N}}\| F_i(q)\| < +\infty$ for all $i\in\cco 1,K\ccf$ (while $F_0$ may be unbounded and admit singularities). Under this assumption, using the bounds on the $F_i$'s and the fact the norm of the velocity is not modified at jump times, it can be proven that there is almost surely a finite number of jumps in a finite interval\cite{DurmusGuillinMonmarcheToolbox}. As a consequence, the trajectory is defined for all times.

It is readily checked\cite{Monmarche2019KineticWalks} that $\int \mathcal L\varphi \dd \mu = 0$ for all nice $\varphi$ where $\mathcal L$ is the Markov generator associated to the process, which informally means that $\mu$ is an invariant distribution for the process. The rigorous proof that the process is ergodic with respect to $\mu$ in the typical settings of molecular dynamics (i.e. with singular forces) is beyond the scope of the present paper, but we give a sketch of the proof of the following:

\begin{thm}
In the case of periodic conditions, i.e. $q\in(\R/\Z)^{3N}$, suppose that $F_i$ is $\mathcal C^2$ for all $i\in\cco 0,K\ccf$. Then the process is ergodic with respect to $\mu$, in the sense that \eqref{EqErgodic} holds for all bounded measurable $\varphi$.
\end{thm}

\begin{proof}
It is well-known and easily checked (see \cite{LeimkuhlerMatthewsStoltz} and references within) that the Langevin diffusion \eqref{Eq-SDE-Langevin-F0} admits $E_{kin}$ as a Lyapunov function. Moreover, by hypoellipticity, its transition kernel admits for all positive time a smooth density with respect to the Lebesgue measure. Since the kinetic energy is unchanged at a jump time, $E_{kin}$ is still a Lyapunov function for the jump process. Since the jump rates are bounded, there is for all positive time $t$ a positive probability $r$ (independent from the starting point of the process) that no jump occurs in the time interval $[0,t]$. Thus the transition kernel of the jump process is bounded below by $r$ times the transition kernel of the diffusion \eqref{Eq-SDE-Langevin-F0}, hence is bounded below (uniformly in the starting point in a given compact set) by a constant times the Lebesgue measure on some compact set. This means a local Doeblin condition is satisfied which, together with the drift condition enforced by the Lyapunov function, classically concludes (see \cite{Monmarche2019KineticWalks,DurmusGuillinMonmarche2018,2018MonmarcheCouplage,LeimkuhlerMatthewsStoltz} for details).
\end{proof}

\subsection{Interpretations}

Let us give some intuitions on the meaning of the jump mechanisms. Suppose to fix ideas that $F_i=\nabla E_i$ for some potential $E_i$, $\beta=1$ and $M$ is the Identity matrix. Then, from $\dot q= v$, we get $\partial_t E_i(q) = v' F_i(q)$, hence $\lambda_i(q,v) = \max(0,\partial_t E_i(q))$. As a consequence,
\[\int_0^t \lambda_i\po q(s),v(s)\pf \dd s \ = \ E_i\po q(t)\pf - E_i\po q(0)\pf  \] 
while $E_i$ is increasing along the trajectory, and zero otherwise. The definition of the time $S_i$ can thus be interpreted as follows: the exponential random variable $\mathcal E_i$ drawn at the beginning of the trajectory is the total amount of energy barrier that the process is allowed to cross uphill. When all this budget is spent, the process jumps  (see Fig. \ref{Fig-jump-time}). Remark that, due to the so-called lack of memory of the exponential law, i.e. the fact that, for all $s,t>0$,
\[\mathbb P \po \mathcal E > t+s\ |\ \mathcal E > t \pf = \mathbb P \po \mathcal E > s\pf \,,\]
it is in fact not necessary to keep in memory the accumulated jump rate involved in \eqref{Eq-Si}. Indeed, if we run the trajectory for some time $t_0>0$ and no jump occurs, we can reset the accumulated jump rate to zero and drew new variables $\mathcal E_1,\dots,\mathcal E_k$, and this won't affect the statistics of the trajectory. In other words, the process is Markovian (at any time, the law of the future trajectory depends on the current position but not on the past trajectory).

\begin{figure}
\centering
\includegraphics[scale=0.4]{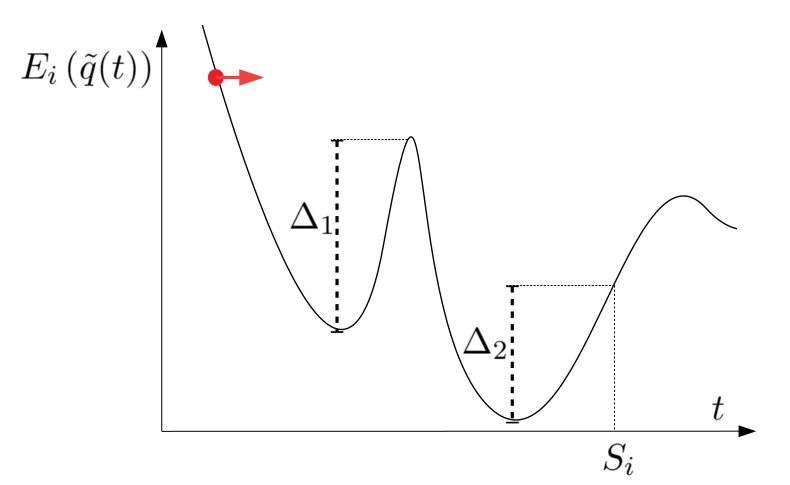}
\caption{As long as the process goes down the potential, there is no jump. When the total energy barrier crossed (here $\Delta_1+\Delta_2$) reaches the random value $\mathcal E_i$, the process jumps.}\label{Fig-jump-time}
\end{figure}

If the process has reached $(q,\tilde v)$ at time $S_i$ and jumps, then $F_i(q)=\nabla E_i(q)$ is the normal vector at point $q$ of the level set $\{E=E(q)\}$. Then the jump from $\tilde v$ to $R_i(q,\tilde v)$   corresponds to a specular reflection on the level set (see Figure \ref{Fig-reflection}).

\begin{figure}
\centering
\includegraphics[scale=0.2]{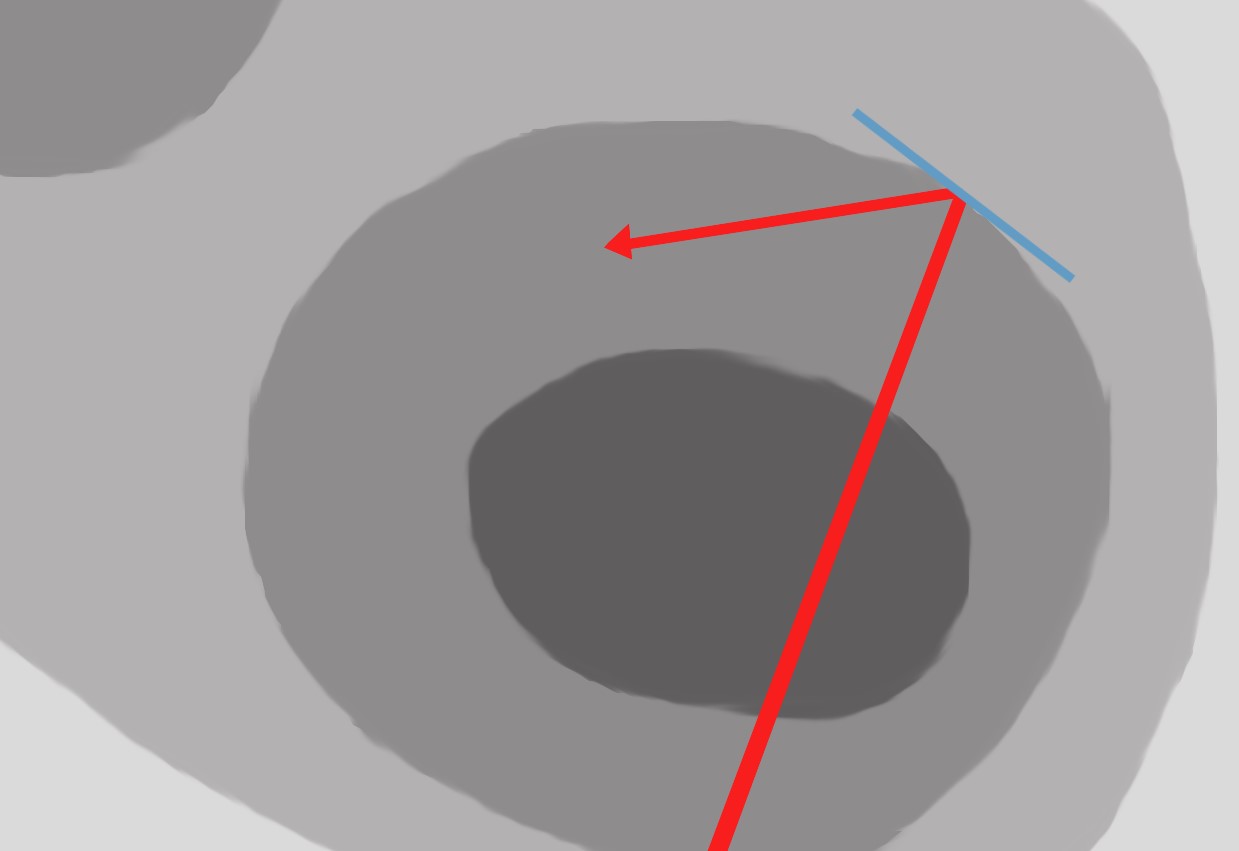}
\caption{At a collision time, the process bounces on the level set of the potential.}\label{Fig-reflection}
\end{figure}

Another interpretation is to see the process as the continuous limit of a lifted Metropolis-Hastings scheme, as in \cite{PetersdeWith,Monmarche2019KineticWalks}. For a given time-step $\delta>0$, consider the Markov transition from $(q_n,v_n)$ to $(q_n+ \delta v_n,v_n)$ with probability $\min( 1,\exp(-E_i(q_n+\delta v_n) + E_i(q_n))$, and to $(q_n+ \delta (v_n+R_i(q_n,v_n))/2,R_i(q_n,v_n)$ otherwise. 
Then this step (approximately as $\delta$ vanishes) fixes the probability law  $\mu$, by design of the Metropolis-Hastings acceptance probability. Remark that
\begin{eqnarray*}
 \min \po 1, e^{-E_i(q+\delta v)+E_i(q)}\pf & =& e^{-\max\po 0,E_i(q+\delta v)-E_i(q)\pf} \\
 &= & e^{- \delta \max\po 0, v'F_i(q)\pf} + o(\delta)\,.
\end{eqnarray*}
Then a bounce for the continuous process corresponds to a Metropolis rejection for the discrete chain. However, contrary to what happens for the classical reversible Metropolis-Hastings algorithm, here a rejection doesn't mean that the process is stopped (which would  impair the exploration of the space, hence the variance of the ergodic estimation), only that it changes its direction. For more complete and rigorous considerations on this interpretation, we refer to \cite{Monmarche2019KineticWalks} and references within.

\subsection{Markov generator and splitting scheme}\label{Sec-Splitting}

For a given observable $\varphi$, denote
\[P_t \varphi(q,p) \ = \ \mathbb E\co \varphi(q_t,p_t)\ |\ (q_0,p_0)=(q,p)\cf \]
the average of $\varphi$ over the distribution of $(q_t,p_t)$. Then this distribution is characterized by $\{P_t\varphi,\ \varphi\text{ continuous bounded}\}$. Similarly, for a Markov process, the motion of the process is characterized by the evolution of $P_t\varphi$ for a sufficiently large class of observables $\varphi$. This evolution is guided by the generator $\mathcal L$ of the process, which is a linear operator defined by
\[\mathcal L \varphi \ = \ \lim_{t\rightarrow 0} \frac{1}t\po P_t \varphi - \varphi\pf \,. \]
 for all suitable observable $\varphi$. In that case, from the Markovian property of the dynamics, we get that $\partial_t P_t \varphi = \mathcal LP_t \varphi$ for all $t\geqslant 0$ so that, informally, $P_t \varphi = e^{t\mathcal L} \varphi$.  Here, for the process introduced in Section \ref{Sec-Deftheorique}, following \cite{Monmarche2019KineticWalks,DurmusGuillinMonmarcheToolbox} (and the notations of \cite{Leimkuhler}) we get that $\mathcal L =  \mathcal L_A + \mathcal L_B + \mathcal L_O + \mathcal L_J $ where 
 \begin{eqnarray*}
 \mathcal L_A \varphi(q,p) & = & (M^{-1} p) \nabla_q \varphi(q,p)\\
  \mathcal L_B \varphi(q,p) & = & -F_0(q) \nabla_p \varphi(q,p)\\
  \mathcal L_O \varphi(q,p) & = & -\gamma (M^{-1} p) \nabla_p \varphi(q,p) + \gamma \beta^{-1} \Delta_p \varphi(q,p)\\
  \mathcal L_J \varphi(q,p) & = & \sum_{i=1}^K \lambda_i(q,M^{-1} p) \po \varphi\po q,p_*\pf - \varphi(q,p)\pf \,,
 \end{eqnarray*}
 with $p_*=MR_i(q,M^{-1})p$,  respectively corresponds to the free transport (A), the forces (B), the friction/dissipation (O), and the jumps (J). We approximate the continuous-time dynamics above following the Trotter/Strang splitting scheme\cite{Leimkuhler,BouRabee} 
\[e^{t \mathcal L} = e^{\frac12 t\mathcal L_B}  e^{\frac12 t\mathcal L_J} e^{\frac12 t\mathcal L_A} e^{t\mathcal L_O} e^{\frac12 t\mathcal L_A} e^{\frac12 t\mathcal L_J} e^{\frac12 t\mathcal L_B} + \underset{t\rightarrow 0}o(t^2)\,.\]
In other words, starting from the BAOAB scheme of \cite{Leimkuhler}, we add two half-time jump steps between the forces and the the transport parts. The process corresponding to each (half) step can be exactly simulated. Let us focus in the next section on the case of $\mathcal L_J$. The stochastic process $(q_t,p_t)$ associated to $\mathcal L_J$ is the following. The position $q_t = q$ is constant, and the velocity $v_t = M^{-1} p_t$ is a Markov chain that jumps from $v$ to $R_i(q,v)$ at rate $\lambda_i(q,v)$ for all $i\in\cco 1,K\ccf$. Remark that computing $\lambda_i(q,v)$ and $R_i(q,v)$ requires to compute $F_i(q)$.

\subsection{Efficient sampling of a Markov chain}\label{Sec-MarkovChain}

In a general abstract framework, let $\lambda_1,\dots,\lambda_K:\R^d\rightarrow \R_+$ be some jump rates and $R_1,\dots,R_K:\R^d\rightarrow\R^d$. Let us construct a Markov chain $X(t)\in\R^d$ that jumps from $x$ to $R_j(x)$ at rate $\lambda_j(x)$ for $j\in\cco 1,K\ccf$. The basic definition of such a chain is the following: starting at $X(0)=x$, let $\mathcal E_1,\dots,\mathcal E_K$ be independent standard exponential random variables and $S_j = \mathcal E_{j} / \lambda_{j}(x)$. Then $T=\min_{j\in\cco 1,K\ccf} S_j$ is the next jump time, and the Markov chain jumps at time $T$ from $x$ to $R_{J}(x)$ where $J$ is the index such that $T=S_J$ (almost surely uniquely defined). Then, the process starts again from this new position.

Suppose that there exists constants $\lambda_1^*,\dots, \lambda_K^*>0$ such that $\lambda_j(x)\leqslant   \lambda_j^*$ for all $x\in\R^d$, $j\in\cco 1,K\ccf$. Then the construction above is equivalent (in the sense that it defines the same Markov chain) to the following. Starting at $X(0)=x$, let $\mathcal E_1,\dots,\mathcal E_K$ be independent standard exponential random variables and $S_j = \mathcal E_{j} /   \lambda_{j}^*$. Let $T=\min_{j\in\cco 1,K\ccf} S_j$, and $J$ be the index such that $T=S_J$. We say that at time $T$, a jump of type $J$ is proposed. Let $U$ be a random variable uniformly distributed over $[0,1]$. If $U \leqslant \lambda_J(x)/ \lambda_J^*$ (which happens with probability $ \lambda_J(x)/\lambda_J^*$) then the chain jumps from $x$ to $R_J(x)$ at time $T$ (we say that the jump is accepted). Otherwise, the chain does not jump, i.e. $X(T)=x$ (we say that the jump is rejected). In both cases, the process starts again from its current position.

From a computational point of view, the main advantage of the second method is that the rate $\lambda_j(x)$ only have to be computed when a jump of type $j$ is proposed.

For a given fixed $t>0$, let $M_j$ be the number of proposed jumps of type $j$ in the time interval $[0,t]$ and $\{T_{k,j}\}_{k\in\cco 1,M_j\ccf}$ be the times of these proposals. Classical properties of the exponential distribution ensures the following facts: $M_j$ follows a Poisson distribution with parameter $\lambda_j^* t$ and, conditionally to $M_j$, the proposal times $\{T_{k,j}\}_{k\in\cco 1,M_j\ccf}$ are independent and uniformly distributed over $[0,t]$. Moreover,  if $i\neq j$, then $M_j$ and $M_i$ are independent and so are $\{T_{k,i}\}_{k\in\cco 1,M_i\ccf}$ and $\{T_{k,j}\}_{k\in\cco 1,M_j\ccf}$. In fact  we don't need to know the values of the proposal times $T_{k,j}$ since the process is constant between two proposals, but we only need to know the order in which the different  jump types are proposed.

As a consequence, we can  sample $X(t)$ as follows. First, draw $M_1,\dots,M_K$ as independent Poisson variables with respective parameters $\lambda_j^* t$ (that are, for each jump type, the total number of jump proposed during the time interval $[0,t]$). Then, draw a jump type $J\in\cco 1,K\ccf$ in such a way that $\mathbb P(J=j) = M_j/(M_1+\dots+M_K)$ for all $j\in\cco 1,K\ccf$. Propose a jump of type $J$, namely: with probability $\lambda_J(x)/\lambda_J^*$, $x$ jumps to $R_J(x)$, otherwise $x$ stays at $x$. Then update the remaining number of type $J$ jump proposals, i.e. set $M_J \leftarrow M_J-1$. Repeat until all the jump proposals have been considered, i.e. $M_1=\dots=M_K=0$. Then  the current position is $X(t)$.

The method is efficient if the upper bounds $\lambda_j^*$ are small, so that the number of jump proposals (and thus the number of computations of $\lambda_j(x)$) are small. Of course, for $\lambda_j^*$ to be small, we need at least the rates to be small, and then the upper-bound to be close to these rates.

\section{A first example}\label{SectionFirstAlgo}

\subsection{Decomposition of the forces}

Here we consider a total potential energy 
\[E_{pot}(q) = E_{elec}(q) + E_{VdW}(q) + E_{bond}(q)\]
with respectively the Coulomb electric, van der Waals, and bond potential.  For the sake of simplicity, in this first example, the jump mechanisms is only used for the long-range van der Waals forces. The van der Waals potential is 
\[E_{VdW}(q) \ = \  \sum_{i\neq j} W(|q_i-q_j|)\]
where $q_i\in\R^3$ is the coordinates of the center of the $i^{th}$ atom, $i\in\cco 1,N\ccf$, and 
\begin{equation}\label{VdW}
W(r) = E_0 \po \po \frac{r_0}{r}\pf^{12} - \po \frac{r_0}{r}\pf^{6}\pf  
\end{equation}
for some parameters $E_0,r_0>0$.  Consider some smooth switching function $\chi$ from $\R_+$ to $[0,1]$ with $\chi(r)=0$ if $r<r'$ and $\chi(r)=1$ if $r>r''$ for some thresholds $0<r'<r''$. 
For $i\in\cco 1,N\ccf$, let $A_i$ be a $(3N)\times 3$ matrix with all coefficients equal to zero except $A(3(i-1)+1,1)=A(3(i-1)+2,2)=A(3(i-1)+3,3)=1$. In other words, if $v=(v_1,\dots,v_N)\in\R^{3N}$ are the velocities of the $N$ atoms, then $A_i' v = v_i \in\R^{3}$. Then, denoting $r_{ij} = |q_i-q_j|$, decompose
\begin{eqnarray*}
\lefteqn{\nabla E_{VdW}(q)  } \\
&=& \sum_{i=1}^N A_i \nabla_{q_i} E_{VdW}(q)\\
& = & \sum_{i=1}^N A_i \sum_{j \neq i}^N \po   1  - \chi_1(r_{ij}) + \chi_1(r_{ij}) \pf  \nabla_{q_i} \po  W(r_{ij})\pf   \\
& = & F_{shortVdW}(q) + \sum_{i=1}^N A_i \sum_{j = 1}^N F_{i,j}(q)
\end{eqnarray*}
with 
\begin{eqnarray*}
F_{shortVdW}(q) & = &    \sum_{i=1}^N A_i \sum_{j \neq i}^N  \po 1 -  \chi(r_{ij}) \pf \nabla_{q_i} \po  W(r_{ij}) \pf\\
F_{i,j}(q) &=&   \chi(r_{ij})  \nabla_{q_i} \po   W(r_{ij})\pf
\end{eqnarray*}
for $j\neq i$, and zero otherwise. Note that we applied the switching function directly at the gradient level and that the short-range van der Waals forces $F_{shortVDW}$ are numerically less intensive than the long-range ones.

  As a conclusion, decompose the forces as
\[\nabla E_{pot}(q) = F_0 + \sum_{i=1}^N A_i \sum_{j=1}^N F_{i,j}(q) \]
with
\[F_0(q) \ = \  \nabla E_{elec}(q) + F_{shortVdW}(q)  + \nabla E_{bond}(q)\,.\]
Remark that $v' A_i F_{i,j} = v_i' F_{i,j}$. It means that the jump mechanisms associated with the vector field $A_i F_{i,j}$ only involves $v_i$. As a consequence, with such a decomposition, a Markov chain $V=(V_1,\dots,V_N)\in\R^{3N}$ with generator given by $\mathcal L_J$ (for fixed positions $q$) is such that the $V_i$'s are independent Markov chains, that can be simulated in parallel. For all $i\in\cco 1,N\ccf$, $V_i$ is a Markov chain that jumps from $v_i$ to $R_{i,j}(v_i)$ at rate $\lambda_{i,j}(v_i)$ for $j\in\cco1,N\ccf$, where
\begin{eqnarray*}
\lambda_{i,j}(v_i) &=& \beta \max\po 0, v_i' F_{i,j}(q)\pf\\ R_{i,j}(v_i)& =& v_i - 2 \frac{v_i'F_{i,j}(q)}{\| F_{i,j}(q)\|^2} F_{i,j}(q)\,.
\end{eqnarray*}
(the mass of an atom is a scalar matrix and thus it disappears in the expression of $R$). According to Section \ref{Sec-MarkovChain}, we need to find an upper bound of the jump rates. Here,
\[\lambda_{i,j} (v_i)  \ \leqslant \ \beta \| v_i\| \| F_{i,j}(q)\|\,.\]
Remark that, at a jump time, the velocity $v_i$ is transformed by an orthogonal reflection, so that its norm is conserved and thus $\|V_i(t)\|$ is constant along the trajectory of the Markov chain. On the other hand, using that $\|\nabla_{q_i}(r_{ij})\|=1$ and $\|\chi\|_\infty = 1$, we get that
\[\| F_{i,j}(q)\| \ \leqslant \  \| \chi(r_{ij}) W'(r_{ij})\| \ \leqslant \ \| \chi W' \|_{\infty}\,. \]
As a consequence, given any bound $L$ of $\| \chi W'\|_{\infty}$, we bound the jump rate by
\[\lambda_{i,j} (q,v)  \ \leqslant \ \lambda_{i,j}^* \ := \ \beta \|v_i\| L\,.\]
We can bound $\chi$ by the indicator function of $[r'',+\infty[$ to analytically derive an explicit bound of $\|\chi W'\|_{\infty}$. Nevertheless, as we saw in Section \ref{Sec-MarkovChain}, the lower is our bound, the more efficient is the algorithm, so instead of using an explicit non-optimal bound it is better to numerically compute $\|\chi W'\|_{\infty}$ once and for all at the beginning of the simulation. This is easily done since $\chi W'$ is a one-dimensional function.

\subsection{The algorithm}\label{SectionBJAOAJB}

Keep the notations of the previous section, and denote  $\delta$ the time-step. The symmetric splitting scheme presented in Section \ref{Sec-Splitting} reads

\begin{enumerate}
\item[(B)] Set $p \leftarrow p - \frac12\delta F_0(q) $.
\item[(J)] Set $p \leftarrow M V(\delta /2)$ where $V(0) = M^{-1} p$ and $V=(V_1,\dots,V_N)$ where for all $i\in\cco 1,N\ccf$, $(V_i(t))_{t\geqslant 0}$ is a Markov chain that jumps from $v_i$ to $R_{i,j}(v_i)$ at rate $\lambda_{i,j}(v_i)$ for all $j\in\cco 1,N\ccf$ (see below).
\item[(A)] Set $q \leftarrow q + \frac12\delta M^{-1} p$.
\item[(O)] Set $p \leftarrow e^{-\gamma \delta M^{-1} } p + \sqrt{\beta ^{-1} \po 1 - e^{-2\gamma \delta M^{-1}  } \pf M  } G$ with $G$ a standard Gaussian random variable. 
\end{enumerate}
Then, repeat steps A, J and B in this order. It remains to describe in detail the step J, which is the following:

\begin{itemize}
\item For $i\in\cco 1,N\ccf$, do
\item \qquad Initialize $v_i = \frac{1}{m_i}p_i$.
\item \qquad Draw $M_i$ a Poisson random variable with parameter $N\beta \|v_i\| L \delta/2$.
\item \qquad For $k=\cco 1,M_i\ccf$, do
\item \qquad \qquad Draw $J$ and $U$ uniformly distributed respectively over $\cco 1,N\ccf$ and $[0,1]$
\item \qquad \qquad If $U<\lambda_{i,J}(v_i) / \lambda_{i,J}^*$, do
\item \qquad\qquad\qquad $v_i \leftarrow R_{i,J}(v_i)$
\item \qquad \qquad end if
\item \qquad  end do
\item end do.
\end{itemize}

Note the following slight modification with respect to the general settings of Section~\ref{Sec-MarkovChain}: we don't generate different Poisson variables $M_{i,j}$ for each $j\in\cco 1,N\ccf$ corresponding to the $N$ jump mechanisms involving $v_i$. The reason is that the bounds $\lambda_{i,j}^*$ are the same for all $j$, and thus we directly sample the total number of jump proposed for the $i^{th}$ velocity (hence the $N$ factor in the parameter of the Poisson variable $M_i$) and, for each of these jump proposals, we chose at random its type $J$ (again, since the bounds are the same, the index $J$ such that $S_J=\min_{j\in\cco 1,N\ccf} S_j$ is uniformly distributed over $\cco 1,N\ccf$).

\subsection{Computational gain}\label{SectionCompGain}

In a usual scheme (for instance, the BAOAB one\cite{Leimkuhler}), at each time-step, $\nabla_q E_{VdW}$ has to be computed, which demands $N(N-1)$ computations of quantities of the form $\nabla_{q_i} \po  W(r_{ij}) \pf$. Let us compare this with  our algorithm.

At each step $F_{shortVdW}$ has to be computed once, which demands $\sum_{i=1}^N N_{n(i)}$ similar computations, where $N_{n(i)}$ is the number of atoms in a neighbour list that contains at least all atoms at distance less than $r''$  of the $i^{th}$ one. Moreover, at each half time-step and for each atom $i\in\cco 1,N\ccf$,  $M_i$ computations are required, whose expected value is $\beta N \|v_i\| L \delta/2$. At equilibrium, $v_i$ is distributed according to an isotropic Gaussian law with  variance $1/(m_i\beta)$ where  $m_i$ is the mass of the $i^{th}$ atom. By the Jensen inequality, the expected value of $\|v_i\|$ is less than the standard deviation of this distribution, $\sqrt{3/(m_i\beta)}$. Thus the average total number of computations of quantities of the form $\nabla_{q_i} \po  W(r_{ij})  \pf$ at each time step is of order
\[\sum_{i=1}^N N_{n(i)} + \sqrt{3\beta} N  L \delta \sum_{i=1}^N m_i^{-1/2}\,. \]
Remark that the two parts of this quantity depend (in an opposite way) of the choice of the radius $r'$ that distinguishes short and long-range interaction. Indeed, as $r'$ increases, $N_{n(i)}$ increases but the Lipschitz constant $L$ decreases. The optimal choice of $r'$ should minimize the sum.

If we suppose that $N_{n(i)}\ll N$ when $N$ becomes large (the other parameters being fixed), which is the case if the volume of the system increases with a constant density, then the cost of computing the short-range forces is negligible with respect to $N^2$. As a consequence, for large $N$, the expected computation gain, in term of number of computations  of quantities of the form $\nabla_{q_i} \po  W(r_{ij})  \pf$, from using bounces instead of computing the full gradient at each step is a factor 
\[ \alpha \ :=\  \sqrt{3\beta}   L \delta \frac1N \sum_{i=1}^N m_i^{-1/2} \,. \]
The parameters of the droplet numerical experiments of Section~\ref{Section:numerique} yield $\alpha = 5.1*10^{-7}$ for a box of 1500 atoms water box using TIP3P model. In other words, concerning the computations of the long-range van der Waals forces, the cost is improved by five orders of magnitude.

Remark that $\alpha N^2$, which is the numerical cost of the jumps at each step, is proportional to the time-step $\delta$. This means that this cost by unit of (simulation) time is in fact independent from the time-step. In other words, the number of jumps proposed in a given time interval is independent from the time-step, which is consistent with the extensivity property of Poisson processes. The time between two jump proposals (and thus two computations of the force) can be seen as a random time-step. This inter-jump time is, by design,  automatically adapted to the Lipschitz norm of the corresponding force (each force being treated independently from the others). However, note that, even though the numerical cost of a given simulation is independent from the time-step (as far as the jump mechanism is concerned), the quality of the result is still impacted by this time-step, since the Trotter splitting induces some error with respect to the continuous-time theoretical process.

\section{The final "BOUNCE" algorithm}

In Section \ref{SectionFirstAlgo}, the algorithm was kept simple for the sake of clarity. We now present the full algorithm that we will hereafter denote as BOUNCE. We will apply the approach to standard biomolecular force fields: AMBER \cite{wang2000well} and CHARMM \cite{CharmmFF} in the context of two types of boundary conditions: droplets where all the interactions are computed in real space and periodic boundary conditions using standard Smooth Particle Mesh Ewald (SPME)\cite{SPME}.

First, two details regarding the van der Waals forces have been eluded in the previous section: in practice, \eqref{VdW} is already multiplied by a switching function (there is no van der Waals interaction between two particles that are too far). In particular, when periodic boundary conditions are applied, the cutoff is such that an atom is never interacting with multiple periodic replicas since the cutoff is smaller than half the size of the box (minimum image convention). For a given particle $i$ the number $N$ in the parameter of the Poisson law that gives the total number $M_i$ of proposed jump  is not the total number of particles but the total number of particles that interact with $i$ through the van der Waals force (and, when a particle $J$ is drawn to determine the type of the jump, $J$ is uniformly drawn between these particles, not all $\cco 1,N\ccf$). Moreover, the parameters $E_0$ and $r_0$ of the interaction depends on the nature of the two particles involved, and thus the cutoff parameter of $\chi$ and the Lipschitz bound $L$ may depend on the different types of particles.

In the droplets simulations, the Coulomb interaction is treated as the van der Waals one: a cutoff of 5 Angstroms determines which interaction falls within short-range and long-range, the short-range being evaluated in a standard way at each time step and the long-range being evaluated with the same jump mechanism as described above, after having determined the proper bounds. A switching function whose characteristics are detailed in SI is applied at the gradient level so that the short-range Coulomb force goes smoothly to zero near the cutoff. Note that due to the slow decay of Coulomb interactions with respect to the interatomic distance, no cutoff is applied to the full Coulomb interaction.

For periodic boundary conditions with SPME, Coulomb interactions are classically divided in a direct and reciprocal part:
\begin{eqnarray*}
E_{elec}(q) &=& \sum_{i\neq j} \frac{\varepsilon_{i,j}}{r_{ij}} \co \rm{erf}(r_{ij}) +\po 1 - \rm{erf}(r_{ij}) \pf\cf  \\
& := & E_{recip}(q) + E_{direct}(q)
\end{eqnarray*}
with $\rm{erf}$ the error function. Following the SPME algorithm, the reciprocal part is computed in the Fourier space with fast Fourier transform. We decompose the direct part as
\begin{eqnarray*}
\lefteqn{\nabla   E_{direct}(q)  \ = } \\
&  & \sum_{i\neq j}  \co \chi(r_{ij}) +\po 1 - \chi(r_{ij}) \pf\cf \nabla \po \frac{\varepsilon_{i,j}}{r_{ij}} \po 1 - \rm{erf}(r_{ij}) \pf \pf \\
& = & F_{shortdir}(q) + \sum_{i\in N}A_i\sum_{j\neq i} \tilde F_{i,j}(q)
\end{eqnarray*}
for some cutoff function $\chi$, like for the van der Waals forces. The long-range parts $\tilde F_{i,j}$ are Lipschitz vector fields that are treated through a jump mechanism.

The short-range van der Waals and Coulomb forces, together with the bond forces $\na E_{bond}$ are all gathered in $F_0$. The last forces that have to be taken into account are the reciprocal forces $\na E_{recip}$, which are many-body forces. In the basic BOUNCE algorithm, they  are also integrated in $F_0$, hence evaluated at each time-step.
 In fact, there would be no theoretical objection to the use of jump mechanisms to treat many body forces but, as we saw, the practical efficiency of the method is related to the capacity to get good (i.e. small) bounds on the forces (in particular, for each type of jump, the average number of proposals per time-step should be less than 1 so that the forces are often not computed at all). This is a non-trivial question in the case of the reciprocal forces. 
Nevertheless, we can also use for them a classical multi-time-step method, the RESPA algorithm \cite{TuckermanRossiBerne}. More precisely, denote Q the succession of steps BJAOAJB as defined in Section \ref{SectionBJAOAJB}, with some timestep $\delta$ and where the jump step J involves both the van der Waals and Coulomb long-range forces. Let $m$ be an even integer. 

Then the evolution of $(q_t,p_t)$ over a timestep $\Delta=m\delta$ is approximated as follows:
\begin{itemize}
    \item Set $p\leftarrow -\frac{\Delta}{2} \na E_{recip}(q)$.
    \item perform $m$ times the step Q with timestep $\delta$.
    \item Set $p\leftarrow -\frac{\Delta}{2} \na E_{recip}(q)$.
\end{itemize}
In other words, a Trotter/Strang splitting scheme is performed with the decomposition
\[e^{\Delta \mathcal L} \simeq  e^{\frac{\Delta}{2} \mathcal L_{recip}} \po e^{\delta \mathcal L_Q}\pf^{m}e^{\frac{\Delta}{2} \mathcal L_{recip}}\]
where
\[\mathcal L_{recip} \varphi(q,p) = -\na E_{recip}(q) \na_p\varphi(q,p)\,.\]
In practice, what we will call BOUNCE-RESPA in the rest of the text is a further standard bonded/non-bonded decomposition where not only the reciprocal space part of the Coulomb interaction is evaluated at an outer bigger time step, but also the rest of the non bonded part of the potential which is not evaluated through the jump mechanism:
\[F_{out} \ = \ F_{shortdir} + F_{shortVdW} - \nabla E_{recip}\]
such that the algorithms reads:
\[e^{\Delta \mathcal L} \simeq  e^{\frac{\Delta}{2} \mathcal L_{out}} \po e^{\delta \mathcal L_Q}\pf^{m}e^{\frac{\Delta}{2} \mathcal L_{out}}\]

The same reasoning can be made for the droplets simulations without the reciprocal space part of the Coulomb energy and by replacing the short-range direct space part of the Ewald sum by the actual short-range part of the total Coulomb energy.
Note that in principle the forces that are evaluated through jumps could be placed at the outer level of the splitting scheme but, due to the extensivity of the jumps, the computational efficiency would not improve, and we would expect the dynamics to be more perturbed by the splitting. As a consequence, the jump step remains a part of $\mathcal L_Q$.

As we see, the RESPA and BOUNCE methods are straightforwardly compatible. Remark that, while the basic BOUNCE algorithm avoids resonance problems since it is a simple discretization (with a unique time-step $\delta$) of an ergodic theoretical process, the BOUNCE-RESPA may, like the original RESPA, suffer from this limitation.
Furthermore, it is straightforward to use an additional splitting of the potential like RESPA1 (with PME) \cite{RESPA1} with the BOUNCE procedure: in this case, $F_{out}$ is further split into a part corresponding to the short-range portion of the non bonded term that is evaluated at an intermediate time step and a part consisting only of the reciprocal space interactions that are evaluated at a larger outer time step.

\section{Assessments of BOUNCE capabilities: numerical experiments}\label{Section:numerique}

To asses the computational capabilities and accuracy of the BOUNCE integrator, we have performed a first implementation in the Tinker-HP massively parallel software for molecular dynamics.\cite{tinkerhp} All tests have been performed using the classical, non-polarizable TIP3P\cite{TIP3P} and SPC\cite{Berendsen_1987} water models and with the CHARMM\cite{CharmmFF} or AMBER\cite{wang2000well} force fields. All bonds are flexible. Computations have been performed on the Occigen supercomputer (CINES, France) using Intel Broadwell processors. The software will be made freely available to academic users within the next release of the Tinker-HP code\cite{tinkerhp,githubtinkerhp,SiteTinkerHP}.
\subsection{Evaluation of direct space  accelerations using BOUNCE: proof of principle droplet simulations}
Because the splitting of the potential between short- and long-range is straightforward in this context, pure direct-space simulations are a good proof of concept to assess the capabilities of this new integrator.  Tests have been performed on large droplets of TIP3P water: first a droplet of 500 TIP3P water molecules and a 9737 atoms solvated ubiquitin protein described with the Amber (f99) force field at 300 Kelvin. In both  cases, a repulsive van der Waals-like wall was imposed to block the atoms to drift away during the simulation. Table I shows the relative acceleration of BOUNCE compared to reference BAOAB (1 fs) simulations and to the popular RESPA with a bonded/non-bonded split (0.5fs/2fs) approach. As expected since BOUNCE evaluates much less often the long-range part of electrostatic and van der Waals interactions than standard integrators, it is nearly 4 times faster than the BAOAB and and twice as fast as BAOAB-RESPA approaches. In fact, when using BOUNCE,  we observe an average number of jump proposals per time-step per atom for the van der Waals interaction of $4.7* 10^{-4}$, in accordance with the a priori estimate of $\alpha*N=7.65*10^{-4}$ given in Section~\ref{SectionCompGain} (corresponding to $\alpha=5.1*10^{-7}$) which explains that BOUNCE always provides a net computational advantage. In that connection, Table \ref{Table-directspace} also shows that the BOUNCE algorithm does adapt to each situation as different force field models give rise to different accelerations, i.e. TIP3P is slightly faster than SPC as is the Amber/TIP3P combination compared to the CHARMM/TIP3P one.A straightforward analysis of the accuracy of the method can be performed by computing the average potential energies of the systems. For the pure TIP3P water droplet after 100ps simulations, the average BAOAB total Potential Energy is of -4371.22 Kcal/mole (-4390.31 Kcal/mole with BAOAB-RESPA 0.5/2fs) whereas BOUNCE obtains a close value of -4387.18 Kcal/mole, with a reference of -4411.80 Kcal/mole for VERLET 0.5fs. In practice the BOUNCE energy is closer to the VERLET reference than BAOAB-RESPA which amounts for -4475.90 Kcal/Mol. Similar results are obtained for a non-homogeneous system such as the solvated ubiquitin (9737 atoms) (see Table \ref{Table-directspace} and Supplementary Information) demonstrating the robustness of the BOUNCE integrator. 
\begin{table}[h!]
\begin{ruledtabular}

\centering
 \begin{ruledtabular}
\centering
\begin{tabular}{||c c c||} 
 Waterbox &   Waterbox & Waterbox   \\
BAOAB & BAOAB-R(0.5fs/2fs) & BOUNCE \\ [0.5ex] 

\hline
\hline
 
 1 (TIP3P) & 1.98 (TIP3P) & 3.96 (TIP3P) \\
 1 (SPC) & 1.95 (SPC) & 3.82 (SPC) \\
  \end{tabular}
\begin{tabular}{||c c c||} 
  Ubiquitin & Ubiquitin & Ubiquitin \\
  BAOAB (1fs) & BAOAB-R(0.5fs/2fs) & BOUNCE\\ [0.5ex] 

\hline
\hline
 
  1 (AMBER) & 1.97 (AMBER) & 3.81 (AMBER) \\
  1 (CHARMM)& 1.95 (CHARMM)& 3.61 (CHARMM) \\
 \end{tabular}
 \end{ruledtabular}

 \caption{Relative speedups for a 1500 molecules TIP3P waterbox and solvated ubiquitin (AMBER FF99/CHARMM protein+TIP3P water) direct-space droplet systems simulation for the BAOAB, BOUNCE and BAOAB-RESPA integrators, with BAOAB as a reference.}
 \label{Table-directspace}
 \end{ruledtabular}
 
\end{table}

\subsection{Accelerating Condensed Phase Simulations: with BOUNCE-RESPA while keeping accuracy}
At this stage, the application of BOUNCE to direct space computations are promising and can already be useful in a variety of cases such as QM/MM computations or continuum solvation models but we want to see how the method performs in the widely used context of periodic boundary conditions using the Particle Mesh Ewald.
 This section proposes an in-depth study of the capabilities of the BOUNCE algorithm to reproduce condensed phase properties of bulk water such as the average potential energies and radial distribution functions. Although the BOUNCE process is not designed to reproduce dynamical properties of the Langevin dynamics, we also include the diffusion coefficient since it is an indicator of the sampling rate.\cite{pushing} The rest of the section shows results on 3 test systems: first a cubic water box of edge 24.662 Angstroms, and then ubiquitin and DHFR proteins solvated in a 72x54x54 Å  water and 64x64x64Å boxes respectively (see simulation details in SI).  

\subsubsection{BOUNCE Computational Performance in Periodic Boundary Conditions: coupling with RESPA/RESPA1}
 We chose BAOAB (1fs) with a friction of 1ps$^{-1}$ and Verlet (0.5) as references and compared BOUNCE to them along with  BOUNCE-RESPA,  BOUNCE-RESPA1 and BAOAB-RESPA1 integrators. The BOUNCE-RESPA and BAOAB-RESPA integrators used a bonded/non-bonded split with a 0.5fs timestep for the bonded forces and a 2fs for the non-bonded ones. The BOUNCE-RESPA1 and BAOAB-RESPA1 use the splitting described above and a 0.5fs timestep for the bonded forces, a 2fs for the intermediate ones and a 6fs timestep for the long-range ones. First, as expected, replacing BAOAB by BOUNCE always results in a net acceleration as the discussed average number of jump proposal per particle per time-step of $5.8* 10^{-3}$ is observed in practice for the direct space interactions. BOUNCE alone without any reciprocal space specific treatment appears competitive with the popular RESPA with a bonded/non-bonded split (0.5/2fs) approach whereas all BOUNCE-RESPA or BOUNCE-RESPA1 outperforms their BAOAB counterparts (see Tables \ref{table-speedupwaterpbc},\ref{table-speedup-ubiquitinpbc}, \ref{table-speedup-dhfrpbc}). It is important to note that by definition, different systems with different model parameters (TIP3P vs SPC, CHARMM vs AMBER) will result in different bounds in the algorithm and therefore BOUNCE will  perform differently among them. It is worth mentioning here too that all proposed results depend on the splitting of the potential used to define the various blocks of the schemes, so that one could choose other cutoffs to define the portion of the potential treated by jumps and the range of the window switch between short- and long-range, that could lead to improved performances. In order to illustrate that, we also compared the performance of simulations using the ubiquitin system and the CHARMM force field with a larger real space cutoff of 10 Å (instead of 7 Å) and a correspondingly larger grid spacing of 1.2 Å (instead of 0.8 Å) as done in other codes\cite{GROMACS}. The results (see Table \ref{table-speedup-ubipbc2}) show an increase of the computational speedup from 1.67 to 1.88 for the pure BOUNCE integrator, 2.22 to 2.33 for the BOUNCE-RESPA and 3.2 to 3.29 for the BOUNCE-RESPA1 integrator compared to the original BAOAB setup.  Finally, another splitting between short-range and long-range for the Coulomb interaction could also be used and could lead to enhanced performances. Since testing all possibilities is beyond the scope of this initial paper, we chose the RESPA1 way of splitting the Coulomb interaction with SPME since its accuracy is less sensitive to the switching function between short- and long-range. \cite{RESPA1switch}.
\subsubsection{BOUNCE accuracy in Periodic Boundary Conditions: computation of condensed phase properties}

We first assessed the capability of the BOUNCE integrator to reproduce the radial distribution functions of water. Figure \ref{Fig-jump-OO} displays the oxygen-oxygen radial distribution in a 1500 atom water-box with the SPC model, Figure \ref{Fig-jump-OH} the Oxygen-hydrogen radial distribution with the same system and model and Figure \ref{Fig-jump-HH} the hydrogen-hydrogen one, after a 2ns simulation. As a reference, we choose a velocity verlet integrator with a small timestep (0.5fs) and the Bussi thermostat\cite{Bussi}.  We compared it to the results obtained with BOUNCE along with  BOUNCE-RESPA,  BOUNCE-RESPA1 and the various BAOAB-based integrators. Figures III to VII show that the native BOUNCE results are extremely close to the reference BAOAB and Verlet values. When coupled with RESPA or RESPA-1, the BOUNCE approaches give a comparable agreement with the reference compared to their BAOAB-RESPA(1) counterparts. Similar trends are observed for average potential energies (see Figures \ref{Fig-waterbox-epot} and \ref{Fig-ubiquitin-epot}).
Knowing that the continuous in time dynamic approached by BOUNCE is different than the Langevin one (given a fixed friction coefficient) which is in turn different than the microcanonical one, we expect dynamical properties computed with these approaches to differ. But the diffusion coefficient of a dynamic is also an indicator of its sampling rate\cite{BerendsenHMR} which makes it interesting to compute: a low self-diffusion coefficient would indicate a lower sampling rate and would then negate the observed computational speedups. Although BOUNCE does indeed affect the dynamics (see discussion in section \ref{Sec-Deftheorique}), the jumps only concern a small fraction of the forces and therefore the self diffusion coefficient (that we computed using the Einstein relation) turns out to be very close to those of the initial Langevin dynamics. In practice, Figure \ref{Fig-jump-diff} shows that the self-diffusion constant of BOUNCE and its RESPA-based variations remains close to the reference BAOAB, meaning that using BOUNCE does not degrade the sampling rate of a dynamic compared to a pure Langevin dynamic computed with BAOAB and with same friction.

\begin{table}[h!]
\begin{ruledtabular}
\centering
 \begin{ruledtabular}
\centering
\begin{tabular}{||c c c||} 
 1fs & 0.5/2fs & 0.5/2/6fs  \\
BOUNCE & BOUNCE-R & BOUNCE- R1  \\ [0.5ex] 

\hline
\hline
 
 1.58 (SPC) & 2.575 (SPC) & 3.25 (SPC) \\
 1.58 (TIP3P) & 2.51 (TIP3P) & 3.18 (TIP3P) \\
 \end{tabular}
\begin{tabular}{||c c c||} 
  & 0.5/2fs & 0.5/2/6fs  \\
 & BAOAB-R & BAOAB-R1\\ [0.5ex] 

\hline
\hline
  & 1.78 (SPC)  &  2.5(SPC) \\
 & 1.78 (TIP3P)  & 2.49 (TIP3P)
 \end{tabular}
 \end{ruledtabular}

 \caption{Speedup comparison on a 1500 atoms periodic boundary conditions waterbox system (SPC and TIP3P water models) of the BOUNCE, BOUNCE-RESPA, BOUNCE-RESPA1, BAOAB-RESPA and BAOAB-RESPA1 integrators, with BAOAB as a reference }
 \label{table-speedupwaterpbc}
 \end{ruledtabular}
\end{table}

\begin{table}[h!]
\begin{ruledtabular}
\centering
 \begin{ruledtabular}
\centering
\begin{tabular}{||c c c||} 
 1fs & 0.5/2fs & 0.5/2/6fs  \\
BOUNCE & BOUNCE-R & BOUNCE- R1  \\ [0.5ex] 

\hline
\hline
 
 1.40 (AMBER) & 1.85 (AMBER) & 2.7 (AMBER) \\
 1.67  (CHARMM) & 2.2  (CHARMM) & 3.2  (CHARMM) \\
 \end{tabular}
\begin{tabular}{||c c c||} 
  & 0.5/2fs & 0.5/2/6fs  \\
 & BAOAB-R & BAOAB-R1\\ [0.5ex] 

\hline
\hline
 
 & 1.64 (AMBER) & 2.55 (AMBER) \\ 
 & 1.65 (CHARMM) & 2.67 (CHARMM)
 \end{tabular}
 \end{ruledtabular}

 \caption{Speedup comparison on a 9737 atoms  periodic boundary conditions ubiquitin protein system (Amber FF99/CHARMM model protein and TIP3P model water) of the BOUNCE, BOUNCE-RESPA, BOUNCE-RESPA1, BAOAB-RESPA and BAOAB-RESPA1 integrators, with BAOAB as a reference}
 \label{table-speedup-ubiquitinpbc}
 \end{ruledtabular}
\end{table}

\begin{table}[h!]
\begin{ruledtabular}
\centering
 \begin{ruledtabular}
\centering
\begin{tabular}{||c c c||} 
 1fs & 0.5/2fs & 0.5/2/6fs  \\
BOUNCE & BOUNCE-R & BOUNCE- R1  \\ [0.5ex] 

\hline
\hline
 
 1.36 & 1.71 & 2.43
 \end{tabular}
\begin{tabular}{||c c c||} 
  & 0.5/2fs & 0.5/2/6fs  \\
 & BAOAB-R & BAOAB-R1\\ [0.5ex] 

\hline
\hline
 
 & 1.57 & 2.21
 \end{tabular}
 \end{ruledtabular}

 \caption{Speedup comparison on a 23558 atoms DHFR   periodic boundary conditions protein system (CHARMM force field and TIP3P Water) of the BOUNCE, BOUNCE-RESPA, BOUNCE-RESPA1, BAOAB-RESPA and BAOAB-RESPA1 integrators, with BAOAB as a reference}
 \label{table-speedup-dhfrpbc}
 \end{ruledtabular}
\end{table}
\begin{table}[h!]
\begin{ruledtabular}
\centering
 \begin{ruledtabular}
\centering
\begin{tabular}{||c c c||} 
 1fs & 0.5/2fs & 0.5/2/6fs  \\
BOUNCE & BOUNCE-R & BOUNCE- R1  \\ [0.5ex] 

\hline
\hline
 
 1.88 & 2.33 & 3.29
 \end{tabular}
\begin{tabular}{||c c c||} 
 1fs & 0.5/2fs & 0.5/2/6fs  \\
 BAOAB & BAOAB-R & BAOAB-R1\\ [0.5ex] 

\hline
\hline
 
 0.79 & 1.33 & 2.35
 \end{tabular}
 \end{ruledtabular}

 \caption{Speedup comparison on a 9737 atoms  periodic boundary conditions Ubiquitin protein system (CHARMM force field and TIP3P Water) of the BOUNCE, BOUNCE-RESPA, BOUNCE-RESPA1, BAOAB, BAOAB-RESPA and BAOAB-RESPA1 integrators using a 10 angstroms real space cutoff and a 1.2 Angstroms grid spacing for PME, with BAOAB and a 7 angstroms real space cutoff and a 0.8 grid spacing for PME (see SI) as a reference}
 \label{table-speedup-ubipbc2}
 \end{ruledtabular}
\end{table}

\begin{figure}
\centering
\includegraphics[scale=0.5]{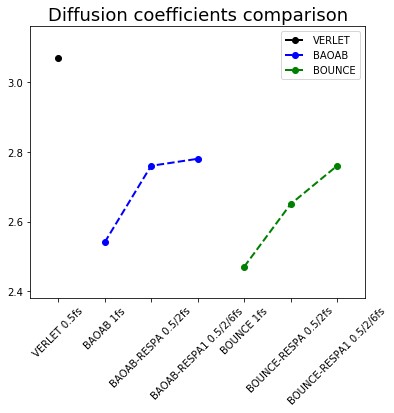}
\caption{Diffusion coefficient in a 1500 atoms  periodic boundary conditions water box using the SPC model after a 2ns simulation. Comparison of the BOUNCE, BOUNCE-RESPA, BOUNCE-RESPA1 and BAOAB-RESPA1 integrators, with BAOAB as a reference.}\label{Fig-jump-diff}
\end{figure}

\begin{figure}
\centering
\includegraphics[scale=0.5]{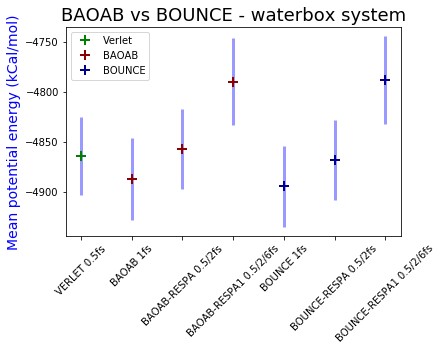}
\caption{Average potential energy comparison in a 1500 atoms  periodic boundary conditions water box using the SPC model after a 2ns simulation. Comparison of the BOUNCE, BOUNCE-RESPA, BOUNCE-RESPA1 and BAOAB-RESPA1 integrators, with BAOAB as a reference.}\label{Fig-waterbox-epot}
\end{figure}

\begin{figure}
\centering
\includegraphics[scale=0.5]{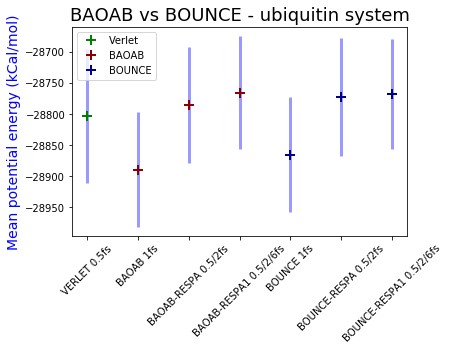}
\caption{Average potential energy comparison in a 9737 atoms  periodic boundary conditions solvated ubiquitin protein system (Amber FF99 protein force field and TIP3P water) after a 1ns simulation. Comparison of the BOUNCE, BOUNCE-RESPA, BOUNCE-RESPA1 and BAOAB-RESPA1 integrators, with BAOAB as a reference.}\label{Fig-ubiquitin-epot}
\end{figure}

\begin{figure}
\centering
\includegraphics[scale=0.1]{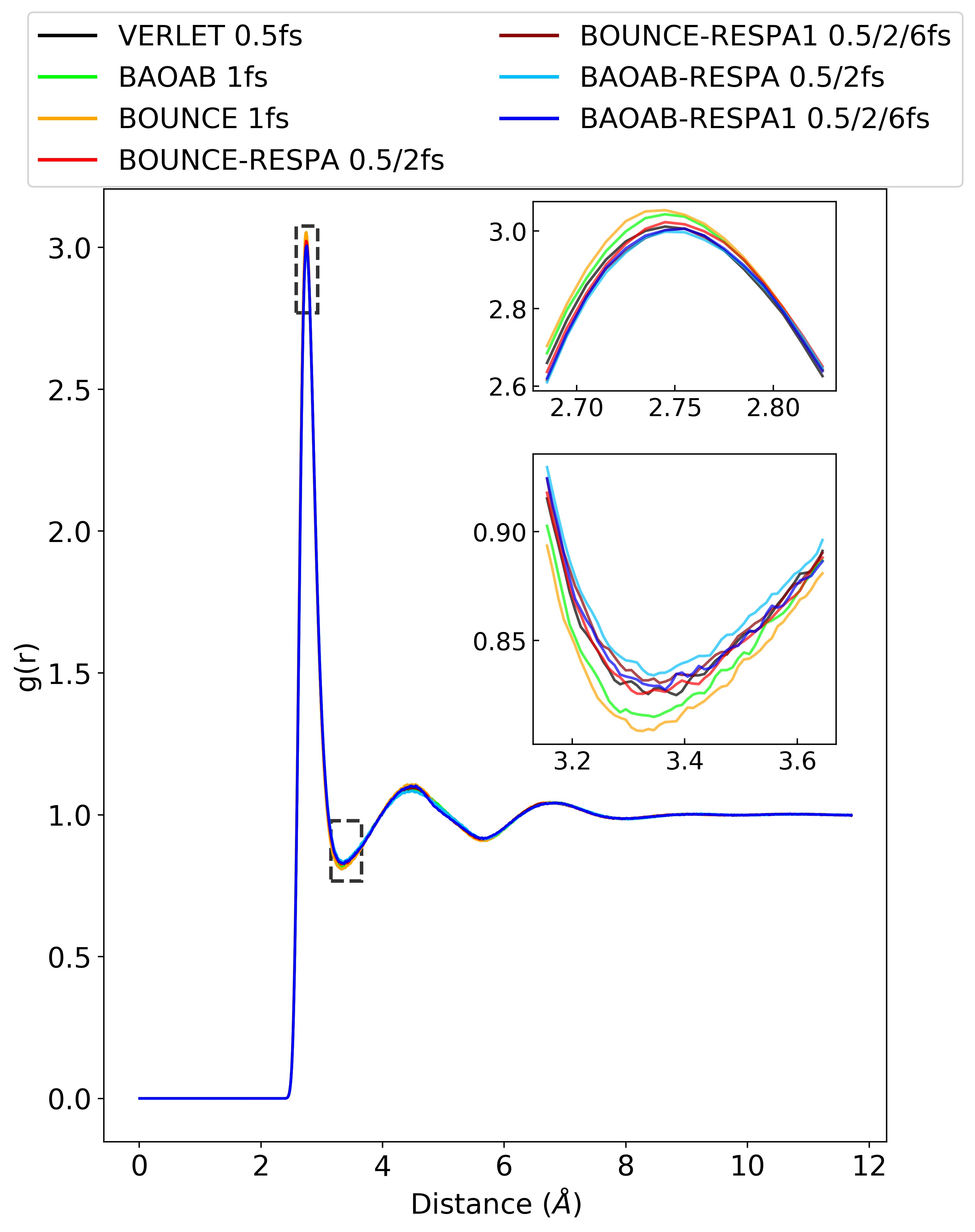}
\caption{Oxygen-oxygen radial distribution obtained with a 1500 atoms  periodic boundary conditions water box using the SPC model after a 2ns simulation. Comparison of the BOUNCE, BOUNCE-RESPA, BOUNCE-RESPA1 and BAOAB-RESPA1 integrators, with BAOAB as a reference.}\label{Fig-jump-OO}
\end{figure}

\begin{figure}
\centering
\includegraphics[scale=0.1]{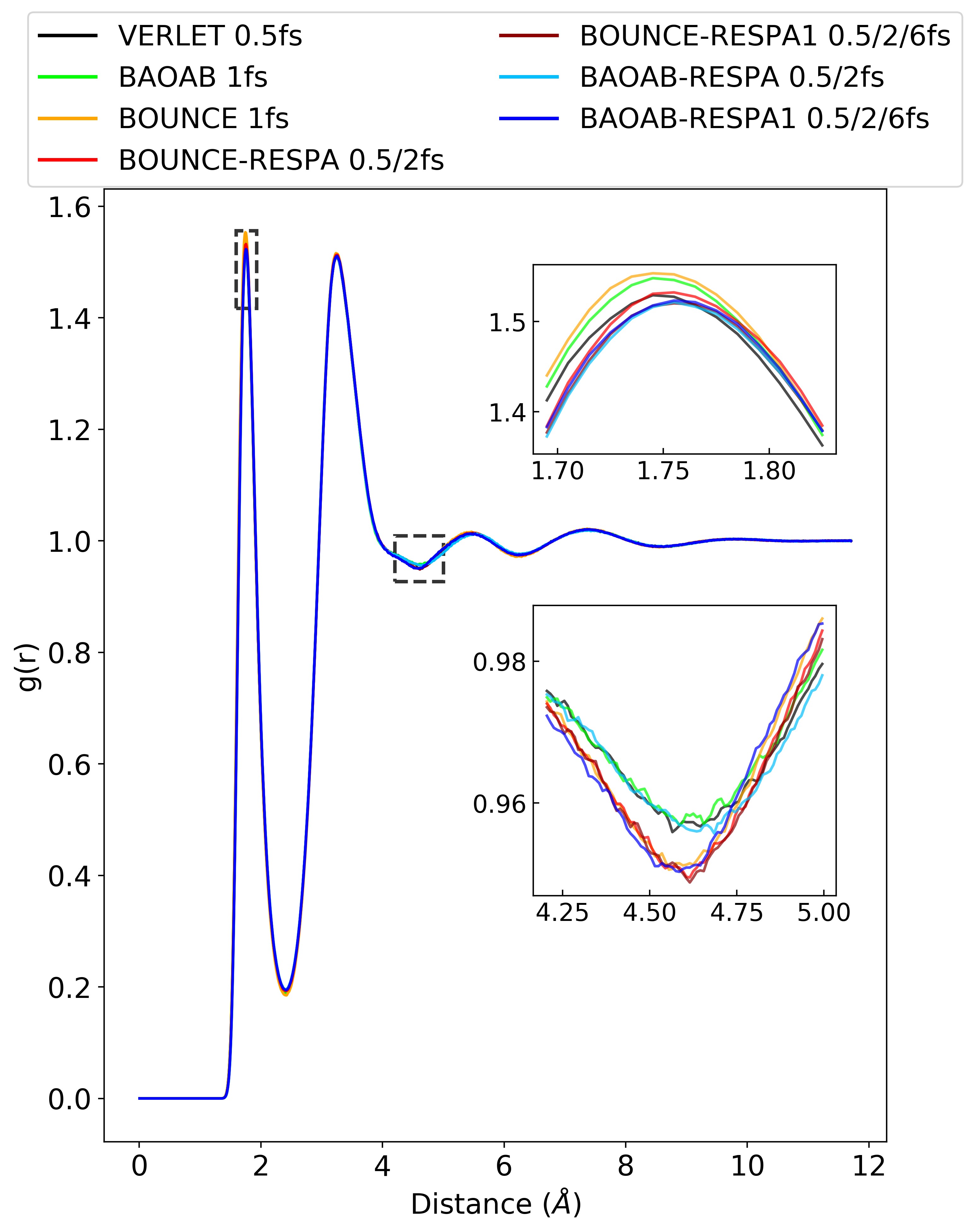}
\caption{Oxygen-hydrogen radial distribution obtained with a 1500 atom  periodic boundary conditions water-box using the SPC model after a 2ns simulation. Comparison of the BOUNCE, BOUNCE-RESPA, BOUNCE-RESPA1 and BAOAB-RESPA1 integrators, with BAOAB as a reference.}\label{Fig-jump-OH}
\end{figure}

\begin{figure}
\centering
\includegraphics[scale=0.1]{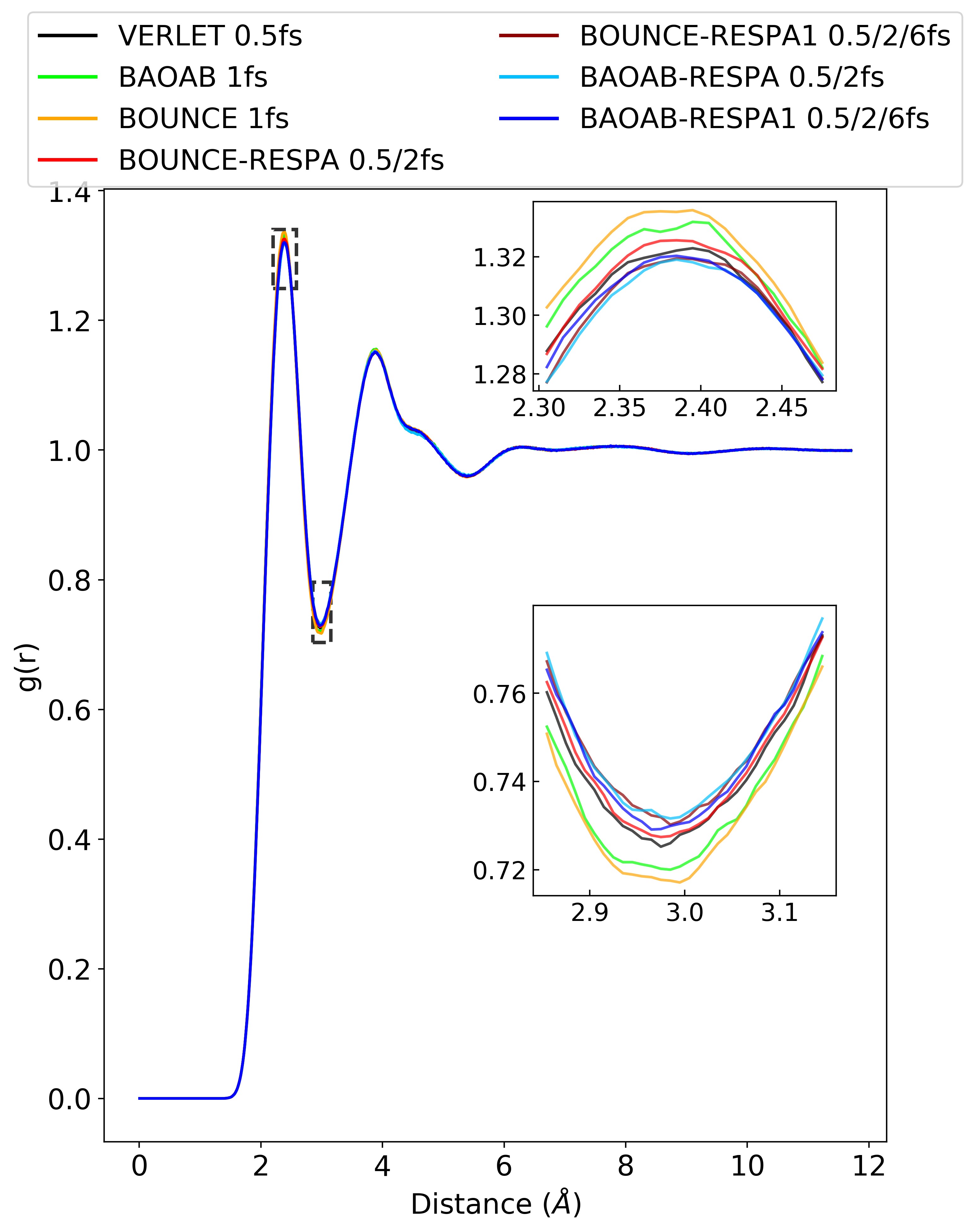}
\caption{Hydrogen-hydrogen radial distribution obtained with a 1500 atom  periodic boundary conditions water box using the SPC model after a 2ns simulation. Comparison of the BOUNCE, BOUNCE-RESPA, BOUNCE-RESPA1 and BAOAB-RESPA1 integrators, with BAOAB as a reference.}\label{Fig-jump-HH}
\end{figure}

\section*{Conclusion and perspectives}
We introduced a new MD integrator, denoted as BOUNCE, in which the slowly varying long-range part of the potential is evaluated through velocity jump processes. By design, the method adapts itself to minimize computational efforts to compute long-range forces while preserving accuracy. In the context of droplet simulations where the non-bonded forces are completely evaluated in direct space, the approach has been shown to have a similar accuracy as a BAOAB integrator with a 1fs timestep while being significantly faster with up to a 4-time acceleration. Significant computational speedup is also obtained in periodic boundary conditions with PME by combining BOUNCE with the RESPA method to evaluate the short-range part of the non bonded forces and reciprocal space contribution less often than the bonded terms: up to a factor 3.2 compared to BAOAB with a 1fs timestep. Furthermore, the use of BOUNCE always computationally outperforms the corresponding BAOAB-RESPA approaches while providing similar accuracy on all static and dynamical computed properties compared to the Verlet or BAOAB reference. Note that in this PBC context, the native basic BOUNCE integrator alone already provides enough direct space acceleration to be computationally competitive with the traditional BAOAB-RESPA (0.5/2fs) integrator. The purpose of the present work has essentially been to present the new algorithm and for this reason, we did not yet applied all acceleration strategies that could be used in conjunction with BOUNCE. For example, as the BOUNCE algorithm strongly minimizes the computation of electrostatics and van der Waals terms in direct space, the choice of any direct summation technique or continuum solvation procedure will strongly benefit from this new integrator. The BOUNCE advantage can also be pushed further in periodic boundary conditions as one could choose to modify the Ewald parameter to increase the computational effort put in direct space compared to reciprocal space (at a fixed Ewald convergence) to benefit more from the BOUNCE efficiency gain as shown in table \ref{table-speedup-ubipbc2}.  Overall, the computational gains obtained with BOUNCE depends on its implementation and on the proportion of the cost of the slowly varying gradient evaluation (which is the one reduced by the algorithm) compared to the rest (neighbor list routines for example).  It is important to point out that the present implementation is only prototypical and the BOUNCE algorithmic structure will provide more room for code optimization\cite{jolly2019raising} whereas further evolutions of the way to treat the many-body reciprocal space interactions are under investigation. Finally, this new technology offers new promises for speeding up simulations using polarizable force fields \cite{melcr2019accurate,pushing,reviewcompchem,annurev-biophys-070317-033349} that include more advanced and therefore more computationally challenging electrostatics, induction and van der Waals contributions that will not need to be computed at each timestep anymore. Incoming works will address all these points in a near future.  
\section{Supplementary Material}
See the Appendix in the supplementary material for details of the simulations and additional details discussed in the text. 
\section{Data availability statement}
The data that support the findings of this study are available from the corresponding author upon reasonable request.
\section{Acknowledgements}
This work has received funding from the European Research Council (ERC) under the European Union's Horizon 2020 research and innovation program (grant agreement No
810367), project EMC2. Computations have been performed at
GENCI on the Occigen machine (CINES, Montpellier, France) on grant no A0070707671.

\bibliographystyle{unsrt}
\bibliography{biblio_jump}

\end{document}